\documentclass[a4paper,onecolumn,11pt,unpublished]{quantumarticle}
\pdfoutput=1
\usepackage[utf8]{inputenc}
\usepackage[english]{babel}
\usepackage[T1]{fontenc}
\usepackage{amsmath, amsthm}
\usepackage{hyperref}

\usepackage{tikz}
\usepackage{lipsum}
\usepackage{pdfpages}

\newtheorem{theorem}{Theorem}
\usepackage{graphicx,amsmath,bm}
\usepackage{physics}
\usepackage{tikz}

\newtheorem{rule-def}[theorem]{Rule}

\begin{document}

\title{An Experimentally Validated Feasible Quantum Protocol for Identity-Based Signature with Application to Secure Email Communication}

\author{Tapaswini Mohanty}
\affiliation{Department of Mathematics, National Institute of Technology Jamshedpur, Jamshedpur-831014, India.}
\email{mtapaswini37@gmail.com}

\author{Vikas Srivastava}
\affiliation{Department of Mathematics, National Institute of Technology Jamshedpur, Jamshedpur-831014, India.}
\email{vikas.math123@gmail.com}

\author{Sumit Kumar Debnath}
\affiliation{Department of Mathematics, National Institute of Technology Jamshedpur, Jamshedpur-831014, India.}
\email{sd.iitkgp@gmail.com}

\author{Debasish Roy }
\affiliation{Department of Mathematics,  IIT Kharagpur, India.}
\email{debasish.roy@maths.iitkgp.ac.in}

\author{Kouichi Sakurai}
\affiliation{Faculty of Information Science and Electrical Engineering, Kyushu University, Fukuoka- 8190395, Japan.}
\email{sakuraicsce2009g@gmail.com}

\author{Sourav Mukhopadhyay}
\affiliation{Department of Mathematics,  IIT Kharagpur, India.}
\email{sourav@maths.iitkgp.ac.in}
	
\maketitle

%%abstract
\begin{abstract}
Digital signatures are one of the simplest cryptographic building blocks that provide appealing security characteristics such as authenticity, unforgeability, and undeniability. In 1984, Shamir developed the first Identity-based signature (IBS) to simplify public key infrastructure and circumvent the need for certificates. It makes the process uncomplicated by enabling users to verify digital signatures using only the identifiers of signers, such as email, phone number, etc.
Nearly all existing IBS protocols rely on several theoretical assumption-based hard problems. Unfortunately, these hard problems are unsafe and pose a hazard in the quantum realm. Thus, designing IBS algorithms that can withstand quantum attacks and ensure long-term security is an important direction for future research.  Quantum cryptography (QC) is one such approach. In this paper, we propose an IBS based on QC. Our scheme's security is based on the laws of quantum mechanics. It thereby achieves long-term security and provides resistance against quantum attacks. We verify the proposed design's correctness and feasibility by simulating it in a prototype quantum device and the IBM Qiskit quantum simulator. The implementation code in qiskit with Jupyternotebook is provided in the Annexure. Moreover, we discuss the application of our design in secure email communication.

\end{abstract}


\markboth{Mohanty, Srivastava, Debnath, Roy, Sakurai, Mukhopadhyay et al. }{A Feasible Quantum Protocol for Identity-Based Signature}

\section{Introduction}
The digital signature is the most commonly utilized cryptographic building block in modern communication networks. The signer generates their public and secret keys in a digital signature protocol. They are distinguished from each other by their public key. In a practical scenario, the signer can be distinguished by their names or email addresses instead of randomly generated keys. Public critical infrastructure (PKI) was intended to create a one-to-one correspondence between public keys and a signer's identity. A certificate authority may be necessitated to connect public keys with identities with the help of a digital certificate. However, this method of assigning public keys with the user's identifier has many loopholes. It lacks effectiveness, and it is not resilient. A more efficient substitute framework is needed to make the PKI simpler. Shamir created the identity-based signature technique (IBS), in which a verifier uses public knowledge about the signer's identity to confirm the signature's authenticity \cite{shamir1984identity}. This approach overcomes the requirement of digital certificates. In an IBS, a trustworthy negotiant or a secret key establisher (SKG) uses the identification to generate the signer's secret key from a master secret key exclusively known to SKG. Many IBS schemes have been proposed \cite{ullah2020lightweight, wei2017forward, ramadan2020identity, zhao2019communication, ko2019forward, song2020efficient, wu2020leakage, wang2020efficient, krzywiecki2019identity, sahana2019provable, james2019efficient} in the area of classical cryptography whose security relies on the number theoretic problems \cite{rivest1978method, kravitz1993digital, koblitz1987elliptic} like factorization problem and discrete logarithm problem. However, these schemes won't be safe in the future as Shor's algorithm \cite{shor1999polynomial} can solve some of these hard problems in polynomial time with the help of an efficient quantum computer.

There is a significant problem in preserving vulnerable information or documents like health records and government documents for a long time. These classically encrypted sensitive documents can be lost during transmission in a public channel. A corrupted party may intercept this sensitive data from the communication channel. This theft of classically encrypted information poses a lot of risk. In the future, when large-scale quantum computers are available, he may try to extract information from the encrypted data. Thus, we need an appropriate solution to secure the confidentiality and integrity of user-sensitive data. While post-quantum cryptography can withstand existing classical and quantum algorithms, it is not a long-term security solution. Introducing new effective or advanced classical or quantum algorithms may break the post-quantum algorithms in the future, making them obsolete.
Quantum cryptography (QC), based on the Heisenberg uncertainty principle and the no-cloning theorem, can solve the aforementioned security risks compared to conventional and post-quantum cryptography. Additionally, because quantum protocols are information-theoretically secure and computationally unattainable, they offer security against quantum computers. Therefore, it is necessary to use QC to construct a digital signature protocol, specifically IBS protocols. \\
This paper is sketched as follows: the following section elaborates on our contribution. Section 3 provides a preliminary background, followed by our proposed design of QIBS. Section 5 provides the correctness of our scheme, with Section 6 providing the security analysis and Sections 7 and 8 providing efficiency and performance analysis. In section 9, we provided a toy example implemented in real hardware, the code of which is given in Annexure. 

\section{Our Contribution}
The conventionally used signatures employed a public key infrastructure model that involved the management of digital certificates. Certificate management is a computationally intensive task that makes PKI-based signatures unsuitable for practical purposes. To tackle this problem, an identity-based signature scheme (IBS) is a workable option. However, nearly all existing IBS protocols rely on several theoretical assumption-based hard problems.  Unfortunately, these difficult challenges are insecure and pose a risk in the quantum realm. Given the current circumstances, creating a quantum-secure IBS is a good idea.
Numerous quantum signature scheme designs are state-of-the-art at the moment \cite{zhang2019high,an2019practical,zhang2020practical,yin2016practical}. However, in the context of QC, there are only a few constructions of IBS \cite{chen2018public,xin2019identity,xin2020efficient,xin2020identity}.
The security of \cite{chen2018public} is also dependent on choosing a one-way function as a random one-way permutation oracle.
Herein, we present the design of a quantum IBS. We utilize the quantum analogue of an OTP (one-time pad) to design an information-theoretically secure quantum IBS. Our proposal consists of three phases: (i) Initializing, (ii) Signing, and (iii) Verification. Using a quantum key distribution protocol, the secret key generator (SKG) provides the signer with a secret key corresponding to their identity during the initializing phase. In addition, the verifier and SKG share a secret key. In the Signing phase, the signer signs with his secret key. Afterwards, during the verification stage, the verifier can confirm the signature with SKG's help. If the implicit encryption is theoretically secure, the designed signature scheme quantum IBS will satisfy unforgeability and undeniability. The unforgeability property ensures that no other party can do the signature on behalf of a signer. In contrast, undeniability ensures that a signer can not deny the signature generation if she did the signature. We assert the suggested strategy obtains long-lived security and stays safe from quantum attacks. The communication and computation cost of our intended design are only $10m+2n$ qubits and  $(23m+3n)\delta+ (3m+n)\beta$ respectively for the message of size $m$ qubits, where $\delta$, $\beta$ and $n$ represent the costs of one straight forward measurement, cost of conversion of a classical bit to a qubit, and bit length of $\phi$ respectively.

Chen et al. \cite{chen2018public} suggested the first quantum IBS using the quantum one-way function, encryption based on identity, and classical security token. However, there is a major drawback in the scheme of Chen et al. \cite{chen2018public}. The design in \cite{chen2018public} involves long-term quantum memory and multiple rounds of quantum swap tests, making it inefficient. Furthermore, unlike our approach, the scheme of \cite{chen2018public} is not impervious to SKG's forgery attack. Our proposed design allows anyone to verify the signature, unlike \cite{chen2018public}. 
In contrast to ideas \cite{chen2018public,xin2019identity,xin2020efficient,xin2020identity}, the suggested approach doesn't use intricate oracle operators. Our scheme achieves less computation cost than the existing quantum IBS schemes \cite{chen2018public,xin2019identity}. The proposed design performs better over \cite{xin2019identity,xin2020identity} from the communication point of view. Since our approach relies only on single photon quantum resources and uses straightforward measurement operators, it is more practical and viable than the schemes of \cite{xin2020efficient,xin2020identity}, which uses Bell states and entangled quantum resources. 
Our design does not use any Oracle one-way function compared to the schemes of \cite{chen2018public,xin2019identity,xin2020efficient,xin2020identity}. Moreover, the suggested scheme supports signature over classical and quantum messages. While the works of \cite{chen2018public,xin2019identity,xin2020efficient,xin2020identity} support only signature over classical message.
A summary of comparable schemes to our suggested design with the existing quantum IBS \cite{chen2018public,xin2019identity,xin2020efficient,xin2020identity} is tabulated in Table \ref{tab:comp}. { We simulated and implemented our scheme in the IBM Qiskit quantum simulator to test the correctness and feasibility of the designed quantum IBS. We implemented our design on a real quantum machine to further validate it. The detailed performance analysis can be found in section \ref{per}. We also investigated that how our proposed scheme can be employed as building block in secure email communication.}

\section{Preliminaries}
\subsection{Quantum One-Time Pad \cite{boykin2003optimal}}\label{qotp}
This part explains the quantum one-time pad (OTP) \cite{boykin2003optimal} procedure between John, the encryptor, and Bob, the decryptor.
There are three algorithms in it: ${\sf KeyGen}$, ${\sf Enc}$ and ${\sf Dec}$ which are described below: \\
In the first step, John and Bob share two randomly chosen secret bits for every qubit. We presume that the sharing of these bits has been done in advance. If the first bit is $0$, then John does nothing. Otherwise, he applies $\sigma_z$ to the qubit. Bob doesn't do anything if the second bit turns out to be $0$, but if the second bit is $1$, then John operates $\sigma_x$. Now John sends the resulting qubit to Bob. The protocol is continued for the rest of the bits. Note that the quantum OTP \cite{boykin2003optimal} is information-theoretically secure.

\begin{description}
	
	\item[$K \leftarrow$ ${\sf KeyGen}$($n$).] On input a positive integer $n$ (length of message),
	%the parameter $\kappa$,
	${\sf KeyGen}$ outputs a $\kappa(\geq 2n)$ bit key $K$ shared between John and Bob using some quantum key distribution protocol.
	%Note that if we are encrypting $n$ qubits, then $\kappa \geq 2n$.
	\item[$E_K(|P\rangle) \leftarrow$ ${\sf Enc}$($|P\rangle,K$).] On input the quantum message  $|P\rangle=\otimes_{i=1}^n |p_i\rangle$ with $|p_i\rangle = a_i|0\rangle+b_i|1\rangle$ and $|a_i|^2+|b_i|^2=1$, and the key $K$, John encrypts the message $|P\rangle$ in the following way:
	\begin{enumerate}
		
		\item Computes $\otimes_{i=1}^{n}Z^{K_{2i-1}} |p_i\rangle$ i.e., if the $(2i-1)$-th bit of $K$ is $1$ then operates the unitary operator $Z$ on the $i$-th qubit of the message  $|P\rangle$; otherwise does nothing, for $i=1,\ldots, n$.
		\item Executes $E_K(|P\rangle)=\otimes_{i=1}^{n} X^{K_{2i}}Z^{K_{2i-1}} |p_i\rangle$ by operating $X$ on the $i$-th qubit of $\otimes_{i=1}^{n}Z^{K_{2i-1}} |p_i\rangle$ if the $(2i)$-th bit of $K$ is $1$ for $i=1,\ldots, n$. 	
	\end{enumerate}
	\item[$|P\rangle \leftarrow$ ${\sf Dec}$($E_K(|P\rangle)$).] Bob decrypts $E_K(|P\rangle)$ by performing the following steps:
	\begin{enumerate}
		\item To get $\otimes_{i=1}^{n}Z^{K_{2i-1}} |p_i\rangle$, operates the unitary operator $X$ on the $i$-th qubit of  $E_K(|P\rangle)=\otimes_{i=1}^{n} X^{K_{2i}}Z^{K_{2i-1}} |p_i\rangle,$ if the $(2i)$-th bit of $K$ is $1$, else does nothing, for $i=1,\ldots, n$. 	
		\item Computes $|P\rangle$ by operating the unitary operator $Z$ on the $i$-th qubit of the $\otimes_{i=1}^{n}Z^{K_{2i-1}} |p_i\rangle$ if the $(2i-1)$-th bit of $K$ is $1$, else does nothing.
	\end{enumerate}
\end{description}
Here $X$ and $Z$ are the unitary operators with matrix representation $\begin{pmatrix}
0 & 1\\ 1 &0
\end{pmatrix}$ and $\begin{pmatrix}
1 & 0\\ 0&-1
\end{pmatrix}$ respectively. Note that the quantum OTP  is information-theoretically secure \cite{boykin2003optimal}.

\subsection{U-Gate \label{ugate}}
The most versatile single-qubit quantum gate is the $U$ gate. It is a parameterized gate with matrix representation given by:
$$U(\theta,\phi,\lambda)=\begin{pmatrix}
\cos(\frac{\theta}{2}) & -e^{i\lambda}\sin(\frac{\theta}{2})\\ e^{i\phi}\sin(\frac{\theta}{2}) & e^{i(\phi+\lambda)}\cos(\frac{\theta}{2})
\end{pmatrix}$$
We take the special case where $\theta=\frac{\pi}{2}$ and $ \lambda=0$. Thus, we have $U(\frac{\pi}{2},\phi, 0)=\frac{1}{\sqrt{2}}\begin{pmatrix}
1 & -1\\ e^{i\phi} & e^{i\phi}
\end{pmatrix}$.

\section{Proposed Quantum Identity-Based Signature}

This section outlines our suggested quantum IBS scheme.

\noindent {\bf A high-level overview:} Our scheme consists of three steps: (i) Initializing phase, (ii) Signing phase, and (iii) Verification phase. Let $Alice_i$ be a signer with identity $ID_i$, and Bob be a verifier.
In the Initializing phase, each signer and a verifier share respective secret keys with the SKG with the help of some quantum key distribution (QKD). During the Signing phase, to sign a quantum message, a signer employs the quantum OTP  (mentioned in Section \ref{qotp}) as a building block. On receiving a pair of message signatures from the signer, the verifier, Bob, communicates with SKG to verify the cogency of the message signature couplet during the verification phase.

\begin{description}
	\item\textbf{Initializing Phase:}
	Let $ID_i=(id_i^1, \dots ,id_i^m)\in \{0,1\}^m $ be the identity of an user $Alice_i$.
	\begin{enumerate}
		
		\item Firstly, making use of a quantum key distribution (QKD) protocol SKG shares a secret key $T_i$ (of size $\geq 2m$) with $Alice_i$ having identity $ID_i$ for the quantum OTP mentioned in Section \ref{qotp}.
		\item In a similar manner, Bob shares a secret key $T_u$ (of size $\geq 2m$)  with SKG for the quantum OTP mentioned in Section \ref{qotp}.
		\item Using $T_i$, Alice shares a secret value $\phi\in (0,2\pi)$ with SKG utilizing quantum authentication process \cite{xin2015quantum}.
		
	\end{enumerate}
	See Figure \ref{initia} for a communication flow in the Initializing phase.
	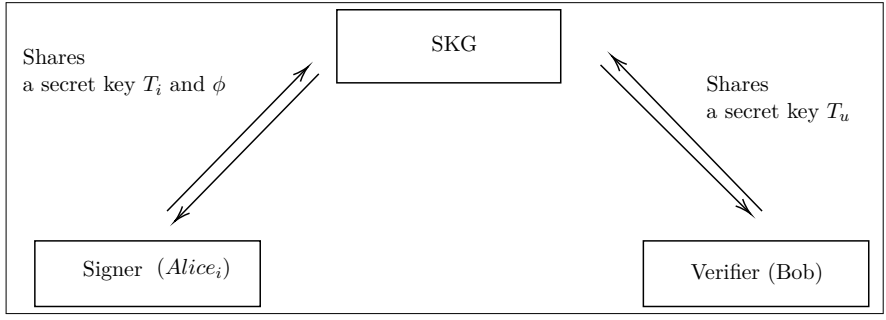
\begin{figure*}[t]
		\centering
		\scalebox{0.75}{
			
			\tikzset{every picture/.style={line width=0.75pt}} %set default line width to 0.75pt
			\fbox{
				\begin{tikzpicture}[x=0.75pt,y=0.75pt,yscale=-1,xscale=1]
				%uncomment if require: \path (0,300); %set diagram left start at 0, and has a height of 300
				
				%Shape: Rectangle [id:dp7492130485072236]
				\draw   (249.5,51) -- (398.5,51) -- (398.5,100) -- (249.5,100) -- cycle ;
				%Shape: Rectangle [id:dp5760130034277359]
				\draw   (48.5,206) -- (198.5,206) -- (198.5,250) -- (48.5,250) -- cycle ;
				%Shape: Rectangle [id:dp8735032215644908]
				\draw   (453.5,206) -- (603.5,206) -- (603.5,250) -- (453.5,250) -- cycle ;
				%Straight Lines [id:da7141806217539451]
				\draw    (237.5,94) -- (142.89,191.56) ;
				\draw [shift={(141.5,193)}, rotate = 314.12] [color={rgb, 255:red, 0; green, 0; blue, 0 }  ][line width=0.75]    (10.93,-3.29) .. controls (6.95,-1.4) and (3.31,-0.3) .. (0,0) .. controls (3.31,0.3) and (6.95,1.4) .. (10.93,3.29)   ;
				%Straight Lines [id:da29829707950411466]
				\draw    (136.5,186) -- (229.11,90.44) ;
				\draw [shift={(230.5,89)}, rotate = 134.1] [color={rgb, 255:red, 0; green, 0; blue, 0 }  ][line width=0.75]    (10.93,-3.29) .. controls (6.95,-1.4) and (3.31,-0.3) .. (0,0) .. controls (3.31,0.3) and (6.95,1.4) .. (10.93,3.29)   ;
				%Straight Lines [id:da39891272605838546]
				\draw    (425,88) -- (523.59,186.59) ;
				\draw [shift={(525,188)}, rotate = 225] [color={rgb, 255:red, 0; green, 0; blue, 0 }  ][line width=0.75]    (10.93,-3.29) .. controls (6.95,-1.4) and (3.31,-0.3) .. (0,0) .. controls (3.31,0.3) and (6.95,1.4) .. (10.93,3.29)   ;
				%Straight Lines [id:da5506483766082935]
				\draw    (532.5,186) -- (433.89,83.44) ;
				\draw [shift={(432.5,82)}, rotate = 46.12] [color={rgb, 255:red, 0; green, 0; blue, 0 }  ][line width=0.75]    (10.93,-3.29) .. controls (6.95,-1.4) and (3.31,-0.3) .. (0,0) .. controls (3.31,0.3) and (6.95,1.4) .. (10.93,3.29)   ;
				
				% Text Node
				\draw (311,67) node [anchor=north west][inner sep=0.75pt]   [align=left] {SKG};
				% Text Node
				\draw (128,216) node [anchor=north west][inner sep=0.75pt]   [align=left] {($\displaystyle Alice_{i})$};
				% Text Node
				\draw (79,218) node [anchor=north west][inner sep=0.75pt]   [align=left] {Signer};
				% Text Node
				\draw (483.5,219) node [anchor=north west][inner sep=0.75pt]   [align=left] {Verifier (Bob)};
				% Text Node
				\draw (39,76) node [anchor=north west][inner sep=0.75pt]   [align=left] {Shares \\a secret key $\displaystyle T_{i}$ and $\displaystyle \phi $};
				\draw (494,94) node [anchor=north west][inner sep=0.75pt]   [align=left] {Shares \\a secret key $\displaystyle T_{u}$};

				\end{tikzpicture}
			}
			
		}
		\caption{Communication flow in Initializing phase}\label{initia}
	\end{figure*}

	\item\textbf{Signing Phase:}
	Let $|P\rangle=\otimes_{j=1}^{m} |p_j\rangle$ be the quantum message of $m$-qubits  that needs to be signed, where $|p_j\rangle = a_j|0\rangle + b_j|1\rangle$ and $|a_j|^2 + |b_j|^2=1$.  $Alice_i$ prepares two copies of $|P\rangle$.
	Given the quantum message  $|P\rangle$, the signer $Alice_i$ generates the signature by performing the following steps:
	
	\begin{enumerate}
		
		\item Operates $U^{\otimes m }$ (where $U=U(\pi/2,\phi,0)$) on the quantum message state $|P\rangle$ to get $U^{\otimes m}(|P\rangle)$.
		
		\item Applies the quantum OTP (as mentioned in Section \ref{qotp}) on the $U^{\otimes m}(|P\rangle)$ to get $| S\rangle =E_{T_i}(U^{\otimes m}(|P\rangle))$.
		
		\item Encodes $ID_i=(id_i^{1}, \dots, id_i^{m})$ into $|ID_i\rangle= \otimes_{j=1}^m |id_{i}^{j}\rangle$.
		
	\end{enumerate}
	Finally, $Alice_i$ outputs the tuple $(| P\rangle, |S\rangle, |ID_i\rangle)$ as message-signature pair. The stages of the Signing phase are depicted in Figure \ref{sign}.

	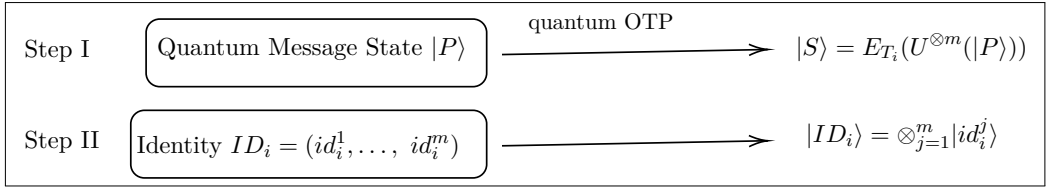
\begin{figure*}[t]
		\centering
		\scalebox{0.85}{
			
			\tikzset{every picture/.style={line width=0.75pt}} %set default line width to 0.75pt
			\fbox{	
				\begin{tikzpicture}[x=0.75pt,y=0.75pt,yscale=-1,xscale=1]
				
				%Rounded Rect [id:dp751550049803263]
				\draw   (129,129) .. controls (129,124.58) and (132.58,121) .. (137,121) -- (321,121) .. controls (325.42,121) and (329,124.58) .. (329,129) -- (329,153) .. controls (329,157.42) and (325.42,161) .. (321,161) -- (137,161) .. controls (132.58,161) and (129,157.42) .. (129,153) -- cycle ;
				%Straight Lines [id:da9483684483587146]
				\draw    (338,142) -- (492,139.04) ;
				\draw [shift={(494,139)}, rotate = 178.9] [color={rgb, 255:red, 0; green, 0; blue, 0 }  ][line width=0.75]    (10.93,-3.29) .. controls (6.95,-1.4) and (3.31,-0.3) .. (0,0) .. controls (3.31,0.3) and (6.95,1.4) .. (10.93,3.29)   ;
				%Rounded Rect [id:dp007940216179514792]
				\draw   (120,183) .. controls (120,178.58) and (123.58,175) .. (128,175) -- (321,175) .. controls (325.42,175) and (329,178.58) .. (329,183) -- (329,207) .. controls (329,211.42) and (325.42,215) .. (321,215) -- (128,215) .. controls (123.58,215) and (120,211.42) .. (120,207) -- cycle ;
				%Straight Lines [id:da8791164690383271]
				\draw    (338,196) -- (492,193.04) ;
				\draw [shift={(494,193)}, rotate = 178.9] [color={rgb, 255:red, 0; green, 0; blue, 0 }  ][line width=0.75]    (10.93,-3.29) .. controls (6.95,-1.4) and (3.31,-0.3) .. (0,0) .. controls (3.31,0.3) and (6.95,1.4) .. (10.93,3.29)   ;
				
				% Text Node
				\draw (352,116) node [anchor=north west][inner sep=0.75pt]   [align=left] {{\small quantum OTP}};
				% Text Node
				\draw (228,140) node [anchor=north west][inner sep=0.75pt]   [align=left] {};
				% Text Node
				\draw (134,130) node [anchor=north west][inner sep=0.75pt]   [align=left] {Quantum Message State $|\displaystyle P\rangle $};
				% Text Node
				\draw (509,128) node [anchor=north west][inner sep=0.75pt]   [align=left] {$\displaystyle |S\rangle =E_{T_{i}}( U^{\otimes m}(|P\rangle) )$};
				% Text Node
				\draw (56,132) node [anchor=north west][inner sep=0.75pt]   [align=left] {Step I};
				% Text Node
				\draw (336,170) node [anchor=north west][inner sep=0.75pt]   [align=left] {};
				% Text Node
				\draw (228,194) node [anchor=north west][inner sep=0.75pt]   [align=left] {};
				% Text Node
				\draw (122,186) node [anchor=north west][inner sep=0.75pt]   [align=left] {Identity $\displaystyle ID_{i} =( id_{i}^{1} ,\dotsc ,\ id_{i}^{m}$)};
				% Text Node
				\draw (56,186) node [anchor=north west][inner sep=0.75pt]   [align=left] {Step II};
				% Text Node
				\draw (515,178) node [anchor=north west][inner sep=0.75pt]   [align=left] {$|\displaystyle ID_{i}$$\displaystyle \rangle =\otimes _{j=1}^{m}$$\displaystyle |id_{i}^{j}$$\displaystyle \rangle $};

				\end{tikzpicture}
			}

		}
		\caption{Stages of Signing phase}\label{sign}
	\end{figure*}

	\item\textbf{Verification Phase:} On receiving a message-signature pair $(| P\rangle, |S\rangle, |ID_i\rangle)$ from $Alice_i$, Bob communicates with SKG to examine the credibility of the message-signature couplet as in the next paragraph:
	%Bob does the following to check the validity of the signature:
	\begin{enumerate}
		
		\item Bob keeps  $| P\rangle$   and encrypts $(|S\rangle, |ID_i\rangle)$ by using the quantum OTP (as mentioned in Section \ref{qotp}) with the help of the secret key $T_u$ to get $E_{T_u}(|S\rangle, |ID_i\rangle)$. Finally, he sends $E_{T_u}(|S\rangle, |ID_i\rangle)$ to SKG.

		\item SKG first decrypts $E_{T_u}(|S\rangle, |ID_i\rangle)$ using the secret key $T_u$ to get back $(|S\rangle, |ID_i\rangle)$. He then retrieves $|P\rangle$  in two steps. In the first step, SKG performs the decryption operation on $|S\rangle$ with the help of the secret key $T_i$, which is associated with the user $Alice_i$ with identity $|ID_i\rangle$. Consequently, SKG gets $U^{\otimes m}(|P\rangle)$. Note that SKG knows secret value $\phi$. In the second step, SKG operates $(U^\dagger)^{\otimes m}$ on $U^{\otimes m}(|P\rangle)$ to retrieve $|P\rangle$. In the following, SKG encrypts $|P\rangle$ using $T_u$  and sends $|R\rangle =E_{T_u}(|P\rangle)$ to Bob.

		\item On receiving $|R\rangle$ from SKG, Bob decrypts it using the secret key $T_u$. Then he verifies whether $D_{T_u}(|R\rangle)$ is same as the $|P\rangle$ of the signature tuple $(| P\rangle, |S\rangle, |ID_i\rangle)$ . If so, then Bob accepts the signature as valid; otherwise, he rejects the signature.
		
	\end{enumerate}
\end{description}
We refer to Figure \ref{veri} for the communication flow in the Verification phase.

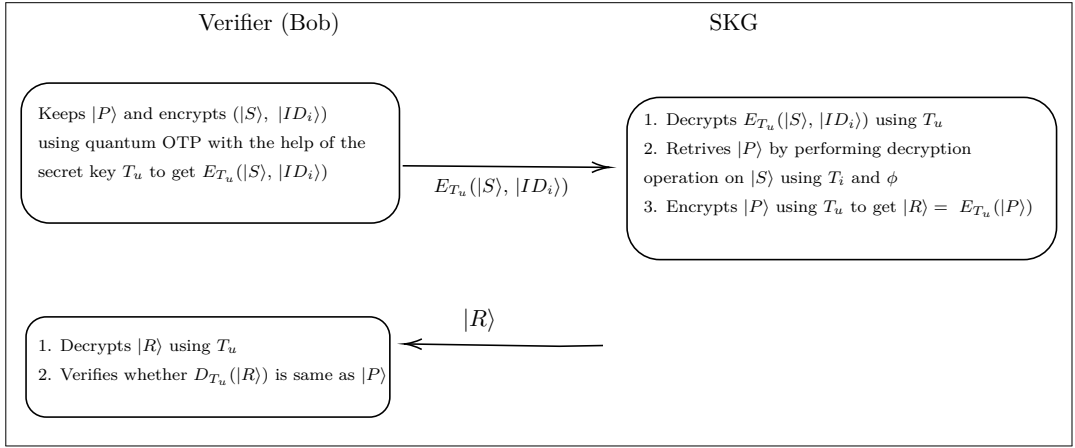
\begin{figure*}[t]
	\centering
	\scalebox{0.8}{

		\tikzset{every picture/.style={line width=0.75pt}} %set default line width to 0.75pt
		\fbox{	
			\begin{tikzpicture}[x=0.75pt,y=0.75pt,yscale=-1,xscale=1]
			%uncomment if require: \path (0,567); %set diagram left start at 0, and has a height of 567
			
			%Rounded Rect [id:dp4952119865683937]
			\draw   (4,226.6) .. controls (4,216.88) and (11.88,209) .. (21.6,209) -- (222.4,209) .. controls (232.12,209) and (240,216.88) .. (240,226.6) -- (240,279.4) .. controls (240,289.12) and (232.12,297) .. (222.4,297) -- (21.6,297) .. controls (11.88,297) and (4,289.12) .. (4,279.4) -- cycle ;
			%Straight Lines [id:da02958928204462352]
			\draw    (242,261) -- (369,261.98) ;
			\draw [shift={(371,262)}, rotate = 180.44] [color={rgb, 255:red, 0; green, 0; blue, 0 }  ][line width=0.75]    (10.93,-3.29) .. controls (6.95,-1.4) and (3.31,-0.3) .. (0,0) .. controls (3.31,0.3) and (6.95,1.4) .. (10.93,3.29)   ;
			%Rounded Rect [id:dp7157789582395314]
			\draw   (382,238.4) .. controls (382,227.13) and (391.13,218) .. (402.4,218) -- (629.6,218) .. controls (640.87,218) and (650,227.13) .. (650,238.4) -- (650,299.6) .. controls (650,310.87) and (640.87,320) .. (629.6,320) -- (402.4,320) .. controls (391.13,320) and (382,310.87) .. (382,299.6) -- cycle ;
			%Rounded Rect [id:dp1211297451471035]
			\draw   (7,368.6) .. controls (7,361.64) and (12.64,356) .. (19.6,356) -- (221.4,356) .. controls (228.36,356) and (234,361.64) .. (234,368.6) -- (234,406.4) .. controls (234,413.36) and (228.36,419) .. (221.4,419) -- (19.6,419) .. controls (12.64,419) and (7,413.36) .. (7,406.4) -- cycle ;
			%Straight Lines [id:da34349306218758546]
			\draw    (367,374) -- (339.85,374.48) -- (245,373.03) ;
			\draw [shift={(243,373)}, rotate = 0.88] [color={rgb, 255:red, 0; green, 0; blue, 0 }  ][line width=0.75]    (10.93,-3.29) .. controls (6.95,-1.4) and (3.31,-0.3) .. (0,0) .. controls (3.31,0.3) and (6.95,1.4) .. (10.93,3.29)   ;
			
			% Text Node
			\draw (113.14,163.22) node [anchor=north west][inner sep=0.75pt]   [align=center] { Verifier (Bob)};
			% Text Node
			\draw (432,164) node [anchor=north west][inner sep=0.75pt]   [align=left] {SKG};
			% Text Node
			\draw (13,222) node [anchor=north west][inner sep=0.75pt]   [align=left] {{\scriptsize Keeps $\displaystyle |P\rangle $ and encrypts $\displaystyle ( |S\rangle ,\ |ID_{i} \rangle )$}\\{\scriptsize using quantum OTP with the help of the}\\{\scriptsize  secret key $\displaystyle T_{u}$ to get $ $$\displaystyle E_{T_{u}}$($\displaystyle |S\rangle $, $\displaystyle |ID_{i}$$\displaystyle \rangle $$\displaystyle )$}\\{\scriptsize  }};
			% Text Node
			\draw (260,267) node [anchor=north west][inner sep=0.75pt]   [align=left] {{\footnotesize $ $$\displaystyle E_{T_{u}}$($\displaystyle |S\rangle $, $\displaystyle |ID_{i}$$\displaystyle \rangle $$\displaystyle )$}};
			% Text Node
			\draw (269,237) node [anchor=north west][inner sep=0.75pt]   [align=left] {};
			% Text Node
			\draw (391,226) node [anchor=north west][inner sep=0.75pt]   [align=left] {{\scriptsize 1. Decrypts $ $$\displaystyle E_{T_{u}}$($\displaystyle |S\rangle $, $\displaystyle |ID_{i}$$\displaystyle \rangle $$\displaystyle )$ using $\displaystyle T_{u}$}\\{\scriptsize 2. Retrives $\displaystyle |P\rangle \ $by performing decryption }\\{\scriptsize operation on $\displaystyle |S\rangle $ using $ $$\displaystyle T_{i}$ and $\phi$}\\{\scriptsize 3. Encrypts $\displaystyle |P\rangle $ using $\displaystyle T_{u}$ to get $\displaystyle |R\rangle =\ E_{T_{u}}( |P\rangle )$}};
			% Text Node
			\draw (13,368) node [anchor=north west][inner sep=0.75pt]   [align=left] {{\scriptsize 1. Decrypts $\displaystyle |R\rangle $ using $\displaystyle T_{u}$}\\{\scriptsize 2. Verifies whether $\displaystyle D_{T_{u}}( |R\rangle) $ is same as $\displaystyle |P\rangle $}\\\\};
			% Text Node
			\draw (278,349) node [anchor=north west][inner sep=0.75pt]   [align=left] { $\displaystyle |R\rangle $ };

			\end{tikzpicture}
			
		}
	}
	\caption{Communication flow in Verification phase}\label{veri}
\end{figure*}

\section{Correctness}

We can easily check the correctness of our scheme in the following way: firstly, SKG can obtain the identity of the signer on receiving $E_{T_u}(|S\rangle, |ID_i\rangle)$ from Bob and using $T_u$ he can decrypt $E_{T_u}(|S\rangle, |ID_i\rangle)$ to get  $(|S\rangle, |ID_i\rangle)$. Now corresponding to the identity $|ID_i\rangle$, SKG will consider $T_i$ and $\phi$ for decrypting $|S\rangle$   to get $|P\rangle$. In the following, SKG encrypts $|P\rangle$ using $T_u$  and sends $|R\rangle =E_{T_u}(|P\rangle)$ to Bob. Then Bob verifies whether $|P\rangle =D_{T_u}(|R\rangle)$ to judge validity of signature. From the above analysis, it is easy to see that the correctness holds in the proposed scheme.

\section{Security Analysis}

\begin{theorem}
	
	If the underlying quantum OTP is information-theoretically secure, the proposed signature scheme quantum IBS satisfies unforgeability and undeniability.
	
\end{theorem}

\begin{description}
	
	\item[\textbf{Unforgeability}]
	
	This property ensures that no other party can do the signature on behalf of a user, say $Alice_i$. Suppose one user $Alice_j$ is different from $Alice_i$ and wants to forge a signature on behalf of the signer $Alice_i$. $Alice_i$'s identity is $|ID_i\rangle$ and $Alice_j$'s identity is $|ID_j\rangle$. $Alice_j$ is having $T_j$ as the shared key and $\phi^{'}$ as the shared secret value between him and SKG.  He has to use $T_j$ and $\phi^{'}$ to encrypt $|P\rangle$ which yields $| S'\rangle =E_{T_j}((U(\pi/2,\phi^{'},0))^{\otimes m}| P\rangle)$. In the following, $Alice_j$ may send $(|P\rangle,|S'\rangle, |ID_i\rangle)$ to Bob who in turn sends $E_{T_u}(|S'\rangle, |ID_i\rangle)$ to SKG. As $T_u$ is known to SKG, he will decrypt $E_{T_u}(|S'\rangle, |ID_i\rangle)$. Since the decryption yields $(|S'\rangle, |ID_i\rangle)$, the SKG uses $T_i$ associated to $|ID_i\rangle$ and $\phi^{'}$ for decrypting $|S'\rangle$. Now the decryption of $|S'\rangle$ using $T_i$ and $\phi$ will produce a value $|P'\rangle$ which is not equal to $|P\rangle$ since $|S'\rangle$ was encrypted  using $T_j$ and $\phi^{'}$.
	Thus, no one can sign on behalf of $Alice_i$, i.e., the proposed scheme achieves unforgeability property.
	
	\item[\textbf{Undeniability}]
	
	This property ensures that a user, say $Alice_i$, can not deny the signature generation if she did the signature. After the signing phase, Bob receives  $(|P\rangle,|S\rangle, |ID_i\rangle)$ as message-signature pair from $Alice_i$. In the following, Bob shares $E_{T_u}(|S\rangle, |ID_i\rangle)$ with SKG who in turn decrypts it and retrieves $(|S\rangle, |ID_i\rangle)$. The SKG then uses $T_i$ corresponding to $|ID_i\rangle$ of $Alice_i$ to decrypt $|S\rangle$. This yields $|P\rangle$, which is sent to Bob for verification. Thus, $Alice_i$ can not deny the generation of the signature if she does this.

\end{description}
\begin{theorem}
The proposed quantum IBS is secure against communicative Pauli operator attack \cite{bibid}.
\end{theorem}
\begin{proof}
 Suppose an attacker wants to forge the signature pair $(| P\rangle, |S\rangle, |ID_i\rangle)$ using the Pauli operator, then he applies a non-trivial Pauli operator $V$ (says), then the signature changes to $(V| P\rangle, V|S\rangle, |ID_i\rangle)$. But $$V|S\rangle\equiv E_{T_u}VU(| P\rangle)\\
                  \neq E_{T_u}U (V(| P\rangle))$$.
Thus, this signature pair is invalid. Therefore, the attacker fails to forge $(| P\rangle, |S\rangle, |ID_i\rangle)$ using the Pauli operator.
\end{proof}

\section{Efficiency Analysis}
We now present the efficaciousness of our proposed blueprint. We first discuss the cost of communication and calculation.
\begin{description}
	
	\item{\bf Communication cost:} $4m$ qubits must be transmitted among the signer, verifier and secret key generator during the Initializing phase. During the Signing and Verification phases, $6m$ more qubits are needed to communicate. If $\phi$ is substituted by a bit string of length $n$, then $2n$ qubits are communicated in the quantum authentication process. So, our proposed design's total quantum communication cost is $10m+2n$ qubits.
	\item{\bf Computation cost:} The entire computational expense of our suggested layout is $(23m+3n)\delta+ (3m+n)\beta$  for the message of size $m$ qubits, where $\delta$ and $\beta$ represent the expenses of converting a conventional bit to a qubit and of making a single, basic measurement, respectively. In particular, $4m+3n$ simple measurements must be performed during the initialising phase. Hence, the cost accrued during the Initializing phase is $(4m+3n)\delta$. In addition, the total cost incurred during encryption and decryption is $19m\delta$. Furthermore, the cost to encode classical bits into quantum bits is $(3m+n)\beta$.
	
\end{description}

We now present the comparative analysis of our proposed design with existing quantum IBS \cite{chen2018public,xin2019identity,xin2020efficient,xin2020identity} in Table \ref{tab:comp}.
\begin{table*}[h!]
	\centering
	\caption{ Comparison with existing quantum IBS}
	\label{tab:comp}
	\scalebox{0.59}{
		
		\begin{tabular}{|p{2.5cm}|p{1.5cm}|l|p{1cm}|p{1.7cm}|p{2cm}|p{3.6cm}|p{3.5cm}|p{3.5cm}|p{1.5cm}|}
			\hline
			& Signature Space & \begin{tabular}[c]{@{}l@{}}Message\\ space\end{tabular} & Key Space & Dimension of Hilbert Space&  Oracle Used & Quantum Communication Cost &  Quantum Computation Cost & Classical Computation Cost & Using Bell States\\ \hline
			Chen et al. \cite{chen2018public} & Q               & C                                                         & C    &2     & Yes                     &  $5m+L_{2}uw$ qubits & $14m\delta+uw \gamma_q +w\gamma_q+uw \delta +w\delta$ & $m\gamma+3mx_2$ & No \\ \hline
			Xin et al. \cite{xin2019identity} & Q               & C                                                         & C     &2    & Yes                     & $11m+8l$ qubits     & $21m\delta+3m\beta+8l(\beta+\delta)$ & $11mx_3+2m\gamma+27m x_2$& No\\ \hline
			Xin et al. \cite{xin2020efficient} & Q               & C                                                         & C      &2   & Yes                     & $8m$ qubits   & $10m\delta+5m\beta$  & $4m\gamma +4mx_1+18mx_2$& Yes\\ \hline
			Xin et al. \cite{xin2020identity} & Q               & C                                                         & C       &2  & Yes                     & $8m+l$ qubits   & $15m\delta+(2m+l)\beta+l\delta$ & $12m\gamma+30mx_2$ &  Yes\\ \hline
			Ours & Q               & Q, C                                                         & C  &2       & No                      & $10m+2n$ qubits & $(23m+3n)\delta+ (3m+n)\beta$    & 0& No \\ \hline
	\end{tabular}}
	\vspace{.1cm}
	
	{\scriptsize
		Q=quantum, C = classical, $l>{}> m$ where $l$ is the total number of decoy particles, $u$=atmost number of signature recipient, $w$=A few safety metric acceptance thresholds utilized in the plan of \cite{chen2018public}, $L_2=$ size of the quantum digest's total qubits, $\delta$ is the computation cost for one basic measurement, $\gamma_q$ is the price of preparing one quantum digital digest, $\beta$ is the price of conversion from a single classical bit into qubit, $x_1$ is computation cost for one way hash function computation, $\gamma$ is the cost of computing one-way function, $x_2$ is computation cost required for one XOR operation, $x_3$ cost incurred during one evaluation of permutation function }
\end{table*}
From the Table \ref{tab:comp}, we can see that unlike \cite{chen2018public,xin2019identity,xin2020efficient,xin2020identity}, our proposed design does not use any oracle one-way functions.  The scheme of \cite{chen2018public} is inefficient as it employs long-term quantum memory and multiple rounds of quantum swap tests. We note that our scheme's computational expense is lower than that of \cite{chen2018public,xin2019identity}. Moreover, our proposed scheme is more communicationally efficient than \cite{xin2019identity,xin2020identity}.
%The scheme of \cite{xin2020efficient} performs slightly better than ours from a quantum communication and quantum computation. However,
%The works [9, 30, 35] need measurement in higher dimensional Hilbert space, \multi-particle
%entangled states" as quantum resources, and \complicated oracle operators". However, with
%the existing technologies implementation of these oracle operators and preparation of these
%resources are not feasible.
The works of \cite{xin2020efficient,xin2020identity} require "entangled states" as quantum facilities and use Bell states. Nevertheless, our plan solely uses single-photon quantum resources, making it more feasible and practical than \cite{xin2020efficient,xin2020identity}.
Since the existing technology is not feasible for preparing these resources. Moreover,  \cite{xin2020efficient,xin2020identity}
is less efficient as it uses Bell states and other complicated oracle operators. Our scheme attains a lower quantum communication cost than the works of \cite{xin2019identity,xin2020identity,chen2018public}.
The messages to be signed in the design of \cite{xin2020efficient,chen2018public,xin2019identity,xin2020identity} must be a classical message.
The scheme of \cite{xin2020efficient,chen2018public,xin2019identity,xin2020identity} can't be used to sign a quantum message because the bits of classical message are used in the signature generation phase
and at some places in the verification phase as well. The fact that \cite{xin2020efficient,chen2018public,xin2019identity,xin2020identity} techniques are limited to signing classical communications is a drawback.  Since our proposed design has no such shortcoming, conventional and quantum messages can be signed.
The key space of
	\cite{xin2020efficient,chen2018public,xin2019identity,xin2020identity} is classical as ours.
 During our scheme's Signing step, the signer uses the conventional secret key $T_i$ to sign the message. In the Verification phase, the verifier Bob uses the classical secret key $T_u$, and SKG uses both $T_u$ and $T_i$ to verify the message. In \cite{xin2020efficient,chen2018public,xin2019identity,xin2020identity}, the signer's secret and public keys are classical bits.
Therefore, our proposed design has the upper hand over all other quantum IBS schemes.
{
	\section{Performance Analysis \label{per}}
	
	We simulated and implemented the proposed quantum IBS protocol on a quantum simulator that runs locally on a classical computer. In the next step, we also tested our design by implementing and running it on an actual back end. 
	\subsection{Implementation on a Quantum Machine \label{sec-imp-1}}
	\noindent We used the quantum simulator provided by the IBM Qiskit, an open-source software development kit, to test the correctness and feasibility of the quantum IBS. The hardware and software specifications are listed in Table \ref{tab:specs}.
	
	\begin{table}[t]
		\centering
		\caption{Software and hardware specification for IBM Qiskit quantum simulation}
		\label{tab:specs}
		\begin{tabular}{|l|l|}
			\hline
			Qiskit & v0.20.2         \\ \hline
			Qiskit Element & Qiskit Aer \\ \hline
			Simulator & QASM \\ \hline
			Python & v3.8.5          \\ \hline
			Local OS     & Linux Lite v5.2 \\ \hline
			Local Hardware & Intel Core i5, RAM 8GB \\ \hline
			
		\end{tabular}
	\end{table}
	
	We verified the correctness of our scheme by taking a particular instance. We take $ID_A={011}$ as the signer's identity. Let $\ket{P}=\ket{010}$ represent the quantum message that requires a signature, $T_i={010110}$ represent the common secret code that links the SKG and the signer, and let $T_U={100101}$ represents the common hidden key between the verifier and SKG. Also, we take $U=U(\pi/2,\pi,0)$.  Figure \ref{fig:quantum-circuit} displays this instance's elaborative circuitry. \ref{fig:quantum-circuit}. 
	\begin{figure*}[t]
		\centering
		\caption{Quantum circuits of an instance of the proposed quantum IBS}
		\fbox{\includegraphics[scale=0.35]{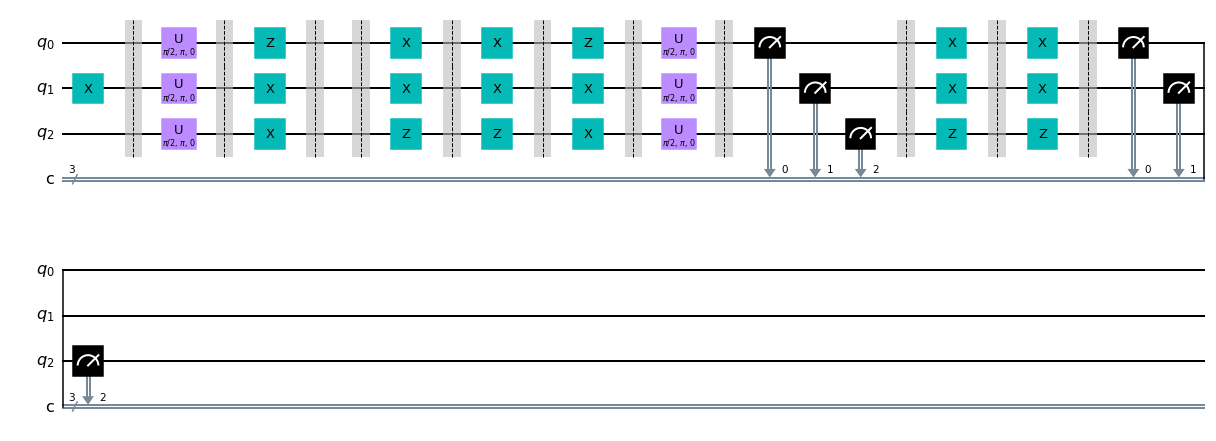}}
		\label{fig:quantum-circuit}
	\end{figure*}
	We also tested the success probability by running an instance of quantum IBS for $10000$ times. The results of this simulation experiment are depicted in Figure \ref{fig:exp}. To conclude, our simulation experiments on the quantum simulator verify the correctness and feasibility of the designed quantum IBS scheme.
	\begin{figure}[h!]
		\centering
		\caption{Success probability of an instance of the proposed quantum IBS}
		\fbox{\includegraphics[scale=0.45]{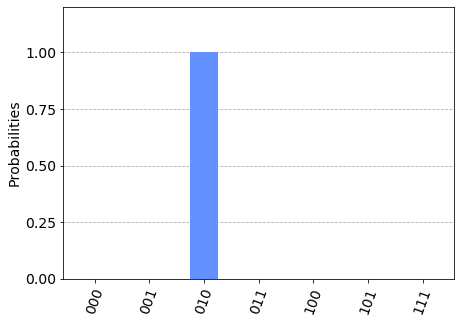}}
		\label{fig:exp}
	\end{figure}
	\subsection{Implementation on a Real Quantum Machine}
	We also implemented the design of quantum IBS put forward in this paper on a real quantum machine provided by IBM Quantum Experience. The hardware and software specifications used for performing simulation experiments are given in Table \ref{tab:specs-real}. We verified the correctness of our proposed protocol by running the same instance of quantum IBS as mentioned in section \ref{sec-imp-1}.

	\begin{table}[t]
		\centering
		\caption{Hardware specification of a real quantum machine used for simulation experiments}
		\label{tab:specs-real}
		\begin{tabular}{|l|l|}
			\hline
			Quantum Computer & IBM Q Lima v1.0.36 \\ \hline
			Qubits & 5 \\ \hline
			Processor type & Falcon r4T         \\ \hline
			CLOPS    & 2.7K \\ \hline
			Quantum Volume & 8 \\ \hline
			Software & Qiskit v0.20.2 \\\hline
			
		\end{tabular}
	\end{table}

	\begin{figure}[t]
		\centering
		\caption{Success probability of an instance of the proposed quantum IBS when simulated on a real quantum machine}
		\fbox{\includegraphics[scale=0.45]{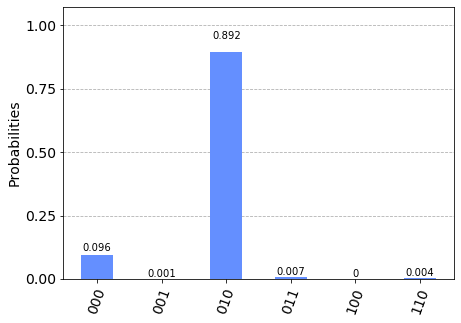}}
		\label{fig:exp1}
	\end{figure}
	
	To find the success probability of the quantum IBS, we ran the protocol for $1024$ times on a real quantum computer with hardware specifications mentioned in Table \ref{tab:specs-real}. Figure \ref{fig:exp1} summarizes the results of this simulation experiment. We now analyze the findings of the experiment. The success probability turns out to be 89.2\%. The error rate of 10.8 \% may be attributed to facts like quantum noise and loss of photons during transmission. For real-life applications, we may increase the number of photons that are being transmitted. In addition, we may use classical error-correction techniques and quantum repeater to avoid the aforementioned problems \cite{shi2022privacy}. In essence, our proposed design only uses single photons quantum resources. As experiments suggest, it is possible and feasible to implement it using the existing quantum back-end technologies.}
\section{Toy Example}

	$X\ket{0}=\ket{1}, X\ket{1}=\ket{0}, Z\ket{0}=\ket{0}, Z\ket{1}=-\ket{1}$
	\large{Encryption}
	\begin{flalign}
	E_k(\ket{P})&=\bigotimes_{i=1}^{3}{X^{k_{2i}}Z^{k_{2i-1}}} \ket{p_i} \\
	&=X^{k_2}Z^{k_1}\ket{p_1}\otimes X^{k_4}Z^{k_3}\ket{p_2} \otimes X^{k_6}Z^{k_5}\ket{p_3} \notag\\
	&=X\ket{p_1} \otimes X\ket{p_2} \otimes Z\ket{p_3} \notag\\
	&=X\ket{0} \otimes X\ket{1} \otimes Z\ket{0} \notag\\
	&=\ket{100} \notag
	\end{flalign}	
	\large{Decryption}: \\
	 $\ket{e}=\ket{100}$
	\begin{flalign}
	Dec(E_k\ket{p})&=\bigotimes_{i=1}^{3}{Z^{k_{2i-1}}X^{k_{2i}}} \ket{e_i}  \\
	&=Z^{k_1}X^{k_2}\ket{e_1}\otimes Z^{k_3}X^{k_4}\ket{e_2} \otimes Z^{k_5}X^{k_6}\ket{e_3} \notag\\	
	&=X\ket{1} \otimes X\ket{0} \otimes Z\ket{0} \notag\\
	&=\ket{010} \notag\\
	&=\ket{p} \notag
	\end{flalign}
\subsection{QIBS}
	\subsubsection{Initialization Phase}
	Let $ID_A={011}$ and $ID_B={100}$ denote the signer and the verifier identities, respectively. $\ket{P}=\ket{010}$ represent the quantum message that requires a signature. Let $T_i={010110}$ denote the secret key between the signer and the SKG, and let $T_U={100101}$ represent the mutual hidden key between the SKG and the person who signed it. Also, we take $U=U(\pi/2,\pi,0)
	= \begin{pmatrix}
	1/\sqrt{2} & -1/\sqrt{2} \\
	-1/\sqrt{2} & -1/\sqrt{2}
	\end{pmatrix} 
	=1/\sqrt{2}\begin{pmatrix}
	1 & -1 \\
	-1 & -1 
	\end{pmatrix}$. Note that
	$U\ket{0}=1/\sqrt{2}\begin{pmatrix}
	1 & -1 \\
	-1 & -1 
	\end{pmatrix}\begin{bmatrix}
	1 \\ 0
	\end{bmatrix}= 1/\sqrt{2}\begin{bmatrix}
	1 \\ -1
	\end{bmatrix} = \ket{-}$ and 
	$U\ket{1}=1/\sqrt{2}\begin{pmatrix}
	1 & -1 \\
	-1 & -1 
	\end{pmatrix}\begin{bmatrix}
	0 \\ 1
	\end{bmatrix}= 1/\sqrt{2}\begin{bmatrix}
	-1 \\ -1
	\end{bmatrix} = \ket{+}$ \\
	\subsubsection{Signing Phase} In the signing phase, the signer having identity $ID_A={011}$, operates $U^{\otimes 3 }$ (where $U=1/\sqrt{2}\begin{pmatrix}
	1 & -1 \\
	-1 & -1 
	\end{pmatrix}$) on the quantum message state $|P\rangle= \ket{010}$ to get $U^{\otimes 3}(|P\rangle)$. In the next step, the signer uses the quantum OTP encryption to encrypt $U^{\otimes 3}(|P\rangle)$ to get $\ket{S}$.
	
	\begin{align}\ket{S}&= E_{T_i}(U^{\otimes 3}(\ket{P})) \notag \\ 
	&=E_{T_i}(U^{\otimes 3}(\ket{010})) \notag \\
	&=E_{T_i}(U\ket{0} \otimes U\ket{1} \otimes U\ket{0}) \notag \\
	&=E_{T_i}(\ket{-} \otimes \ket{+} \otimes \ket{-}) \notag \\
	&=E_{T_i}(\ket{-+-}) \notag \\
	&=E_{T_i}[(\ket{0}-\ket{1}) \otimes (\ket{0}+\ket{1}) \otimes (\ket{0}-\ket{1}) ] \notag \\
	&=E_{T_i}[\ket{000}-\ket{100}+\ket{010}-\ket{110}\\
	&-\ket{001}+\ket{101}-\ket{011}+\ket{111}] 
	\end{align}
	To compute the final value of $\ket{S}$, we compute each of \\$E_{T_i}\ket{000}$, $E_{T_i}\ket{100}$, $E_{T_i}\ket{010}$, $E_{T_i}\ket{110}$,$E_{T_i}\ket{001}$, $E_{T_i}\ket{101}$, $E_{T_i}\ket{011}$, and $E_{T_i}\ket{111}$.
	\begin{align}
	&E_{T_i}\ket{000}=X\ket{0} \otimes X\ket{0} \otimes Z\ket{0} = \ket{110} \notag \\
	&E_{T_i}\ket{100}=X\ket{1} \otimes X\ket{0} \otimes Z\ket{0} = \ket{010} \notag \\
	&E_{T_i}\ket{010}=X\ket{0} \otimes X\ket{0} \otimes Z\ket{1} = \ket{100} \notag \\
	&E_{T_i}\ket{110}=X\ket{1} \otimes X\ket{1} \otimes Z\ket{0} = \ket{000} \notag \\
	&E_{T_i}\ket{001}=X\ket{0} \otimes X\ket{0} \otimes Z\ket{1} = -\ket{111} \notag \\
	&E_{T_i}\ket{101}=X\ket{1} \otimes X\ket{0} \otimes Z\ket{1} = -\ket{011} \notag \\
	&E_{T_i}\ket{011}=X\ket{0} \otimes X\ket{1} \otimes Z\ket{1} = -\ket{101} \notag \\
	&E_{T_i}\ket{111}=X\ket{1} \otimes X\ket{1} \otimes Z\ket{1} = -\ket{001} \notag \\
	\end{align}
	Putting the values back in the original expression ($A.1)$, we get the value of $\ket{S}$ as $\ket{S}= \frac{1}{2\sqrt{2}}[\ket{110}-\ket{010}+\ket{100}-\ket{000}+\ket{111}-\ket{011}+\ket{101}-\ket{001}]$
	
	\noindent Finally, the signer outputs the tuple $(| P\rangle, |S\rangle, |ID_A\rangle)$ as message-signature pair, where $\ket{P}=\ket{010}, \ket{S}= \frac{1}{2\sqrt{2}}[\ket{110}-\ket{010}+\ket{100}-\ket{000}+\ket{111}-\ket{011}+\ket{101}-\ket{001}],$ and  $\ket{ID_A}= \ket{011}$.
	
	\subsubsection{Verification Phase}
	
	\begin{description}

		\item[Step 1] Bob, the verifier, keeps $\ket{P}$ and encrypts $(\ket{S}, \ket{ID_A})$ with $T_u$.
		\begin{align*}
		E_{T_u}&=\bigotimes_{i=1}^{3}{X^{k_{2i}}Z^{k_{2i-1}}} \ket{y_i} \\ \notag
		&=X^{k_2}Z^{k_1}\ket{y_1}\otimes X^{k_4}Z^{k_3}\ket{y_2} \otimes X^{k_6}Z^{k_5}\ket{y_3} \notag\\
		&=Z\ket{y_1} \otimes X\ket{y_2} \otimes X\ket{y_3} \notag\\
		\end{align*}
		$$E_{T_u}\ket{110} =Z\ket{1} \otimes X\ket{1} \otimes X\ket{0} = -\ket{101}  $$
		$$E_{T_u}\ket{010} =Z\ket{0} \otimes X\ket{1} \otimes X\ket{0} = \ket{001} $$
		$$E_{T_u}\ket{100} =Z\ket{1} \otimes X\ket{0} \otimes X\ket{0} = -\ket{111} $$
		$$E_{T_u}\ket{000} =Z\ket{0} \otimes X\ket{0} \otimes X\ket{0} = \ket{011} $$
		$$E_{T_u}\ket{111} =Z\ket{1} \otimes X\ket{1} \otimes X\ket{1} = -\ket{100} $$
		$$E_{T_u}\ket{011} =Z\ket{0} \otimes X\ket{1} \otimes X\ket{1} = \ket{000} $$
		$$E_{T_u}\ket{101} =Z\ket{1} \otimes X\ket{0} \otimes X\ket{1} = -\ket{110} $$
		$$E_{T_u}\ket{001} =Z\ket{0} \otimes X\ket{0} \otimes X\ket{1} = \ket{010} $$
		{ $E_{T_u}\ket{S}= \frac{1}{2\sqrt{2}}[-\ket{110}-\ket{010}-\ket{100}-\ket{000}-\ket{111}-\ket{011}-\ket{101}-\ket{001}] $\\
		$E_{T_u}\ket{S}= -\frac{1}{2\sqrt{2}}[\ket{110}+\ket{010}+\ket{100}+\ket{000}+\ket{111}+\ket{011}+\ket{101}+\ket{001}]=\ket{T} $} \vspace{1mm}
		$$E_{T_u}\ket{ID_A}=E_{T_u}\ket{011}=\ket{000}=EID_A$$
		$E_{T_u}(\ket{S},\ket{ID_A})=(-\frac{1}{2\sqrt{2}}[\ket{110}+\ket{010}+\\\ket{100}+\ket{000}+\ket{111}+\ket{011}+\ket{101}+\ket{001}],\ket{000})$
		This is sent to SKG. For the sake of brevity, we call $E_{T_u}(\ket{S},\ket{ID_A})$ as $ \psi $.

		\item[Step 2:] SKG first decrypts $\psi$ using the secret key $T_u$ to get back $(|S\rangle, |ID_i\rangle)$. 

			\begin{align*}
		D_{T_u}(\ket{\psi}) &= (D_{T_u}\ket{T}, D_{T_u}\ket{EID_A}) \notag \\
		D_{T_u} &=\bigotimes_{i=1}^{3}{Z^{k_{2i-1}}X^{k_{2i}}} \ket{y_i} \notag\\
		&= Z\ket{y_1} \otimes X\ket{y_2} \otimes X\ket{y_3} \notag \\
		D_{T_u}(EID_A) &= D_{T_u}\ket{000}=\ket{011}=\ket{ID_A} \nonumber \\
		D_{T_u}(\ket{T})&=D_{T_u}[-\frac{1}{2\sqrt{2}}[\ket{110}+\ket{010}+\ket{100}+\ket{000}+\\&\ket{111}+\ket{011}+\ket{101}+\ket{001}],\ket{000}] \nonumber\\
		&=\frac{1}{2\sqrt{2}}[\ket{110}-\ket{010}+\ket{100}-\ket{000}+\\&\ket{111}-\ket{011}+\ket{101}-\ket{001}] \nonumber \\
		&=\ket{S}             
		\end{align*}

		\noindent Now SKG performs the following steps to retrieve $\ket{P}$

		\begin{flalign*}
		D_{T_i} &=U^{{\dagger}^{\otimes3}}\bigotimes_{i=1}^{3}{Z^{k_{2i-1}}X^{k_{2i}}} \notag \\
		\bigotimes_{i=1}^{3}{Z^{k_{2i-1}}X^{k_{2i}}} &= X \otimes Z \otimes Z \notag \\
		U &=1/\sqrt{2}\begin{pmatrix}
		1 & -1 \\
		-1 & -1 
		\end{pmatrix} \notag \\
		U^{\dagger} &=1/\sqrt{2}\begin{pmatrix}
		1 & -1 \\
		-1 & -1 
		\end{pmatrix}
		\end{flalign*}
  \begin{widetext}
  \begin{equation}
		U^{{\dagger}^{\otimes3}} = \frac{1}{2\sqrt{2}}\begin{pmatrix}
		1 & -1 & -1 & 1 & -1 & 1 & 1 & -1 \\
		-1 & -1 & 1 & 1 & 1 & 1  & -1 & -1 \\
		-1 & 1 & -1 & 1 & 1 & -1 & 1 & -1 \\
		1 & 1 & 1 & 1 & -1 & -1 & -1 & -1 \\
		-1 & 1 & 1 & -1 & -1 & 1 & 1 & -1 \\
		1 & 1 & -1 & -1 & 1 & 1 & -1 & -1 \\
		1 & -1 & 1 & -1 & 1 & -1 & 1 & -1 \\
		-1 & -1 & -1 & -1 & -1 & -1 & -1 & -1 \\
		\end{pmatrix}
 \end{equation}
\end{widetext}			
		
		$\bigotimes_{i=1}^{3}{Z^{k_{2i-1}}X^{k_{2i}}}\ket{S}= \frac{1}{2\sqrt{2}}[\ket{000}-\ket{100}+\ket{010}-\ket{110}\\-\ket{001}+\ket{101}-\ket{011}-\ket{111}] =\ket{v} (say)$
	
\begin{widetext}
  \begin{equation}
	{\small 	
		U^{{\dagger}^{\otimes3}}\ket{v} = \\\frac{1}{2\sqrt{2}} X \frac{1}{2\sqrt{2}}\begin{pmatrix}
		1 & -1 & -1 & 1 & -1 & 1 & 1 & -1 \\
		-1 & -1 & 1 & 1 & 1 & 1  & -1 & -1 \\
		-1 & 1 & -1 & 1 & 1 & -1 & 1 & -1 \\
		1 & 1 & 1 & 1 & -1 & -1 & -1 & -1 \\
		-1 & 1 & 1 & -1 & -1 & 1 & 1 & -1 \\
		1 & 1 & -1 & -1 & 1 & 1 & -1 & -1 \\
		1 & -1 & 1 & -1 & 1 & -1 & 1 & -1 \\
		-1 & -1 & -1 & -1 & -1 & -1 & -1 & -1 \\
		\end{pmatrix} \begin{bmatrix}
		1 \\
		-1 \\
		1 \\
		-1 \\
		-1 \\
		1 \\
		-1 \\
		1 
		\end{bmatrix}
		}
  \end{equation}
\end{widetext}	
		\begin{flalign*}
		&= \frac{1}{8}\begin{bmatrix}
		0 &0 &8 &0 &0 &0 &0 & 0 
		\end{bmatrix}^{t} \notag \\
		&= \begin{bmatrix}
		0 &0 &1 &0 &0 &0 &0 & 0 
		\end{bmatrix}^{t} = \ket{010}
		\end{flalign*}
		
		In the following, SKG encrypts $|P\rangle$ using $T_u$  and sends $|R\rangle =E_{T_u}(|P\rangle)$ to Bob.
		$E_{T_u}(\ket{P})= E_{T_u} \ket{010} =\ket{001} = \ket{R} $ \\ \\
		
		\item[Step 3] On receiving $|R\rangle$ from SKG, Bob decrypts it using the secret key $T_u$. That is, $ D_{T_u}(\ket{R})= Z\ket{r_1} \otimes X\ket{r_2} \otimes X\ket{r_3} = \ket{010}$. Note that $D_{T_u}(|R\rangle)$ is same as the $|P\rangle$ of the signature tuple $(| P\rangle, |S\rangle, |ID_i\rangle)$. Therefore, Bob accepts the signature as a valid signature.

	\end{description}

\section{Application of Our Proposed Quantum IBS For Secure Email}

Email has become a ubiquitous form of communication, simplifying the exchange of information globally.  However, its widespread usage exposes it to various security threats, including email spoofing and tampering. In this section, we explore the potential usage of IBS in mitigating email security risks, focusing on its ability to prevent spoofing and tampering. Email spoofing is an adversary's common technique to impersonate legitimate senders, deceiving recipients into believing that the email originates from a trusted source.  It can lead to various forms of cybercrime, including phishing attacks, malware distribution, and financial fraud.  Moreover, email tampering involves unauthorized modification of email content during transit, compromising its integrity and authenticity. 

Our proposed quantum identity-based signatures offer a robust solution to address the aforementioned security challenges associated with email communications.  Unlike traditional public key infrastructure (PKI), quantum IBS enables users to generate digital signatures using their identities, such as email addresses or user IDs.  It eliminates the need for a separate infrastructure to manage public keys, simplifying key management processes and enhancing scalability. Quantum IBS can prevent email spoofing by enabling senders to sign their messages using their identities.  When a sender sends an email, they generate a digital signature using their private key.  Recipients can verify the email's authenticity by validating the signature using the sender's public key, derived from their identity.  If the signature is valid, it confirms that the email was indeed sent by the claimed sender, mitigating the risk of impersonation and spoofing.  In addition to preventing spoofing, quantum IBS ensures the integrity of email messages by detecting any unauthorized modifications made to the content during transit.  When a sender signs an email using their private key, the signature encapsulates the entire message, including its content and attachments.  Any alteration to the message would invalidate the signature upon verification.  Thus, recipients can detect tampering attempts and reject emails with invalid signatures, preserving the integrity and authenticity of the communication.  While traditional cryptographic algorithms, such as RSA and ECC, are widely used to secure email communications, they are vulnerable to quantum attacks due to the advent of quantum computing.  On the other hand, our proposed quantum IBS is immune to the threat of quantum computers.  To summarize, quantum IBS presents itself as a robust solution for ensuring that email messages are sent and received securely in the quantum world. Refer to Figure \ref{fig:ibs-email} for a schematic diagram. 
\begin{figure}
    \centering
    \includegraphics[scale=0.12]{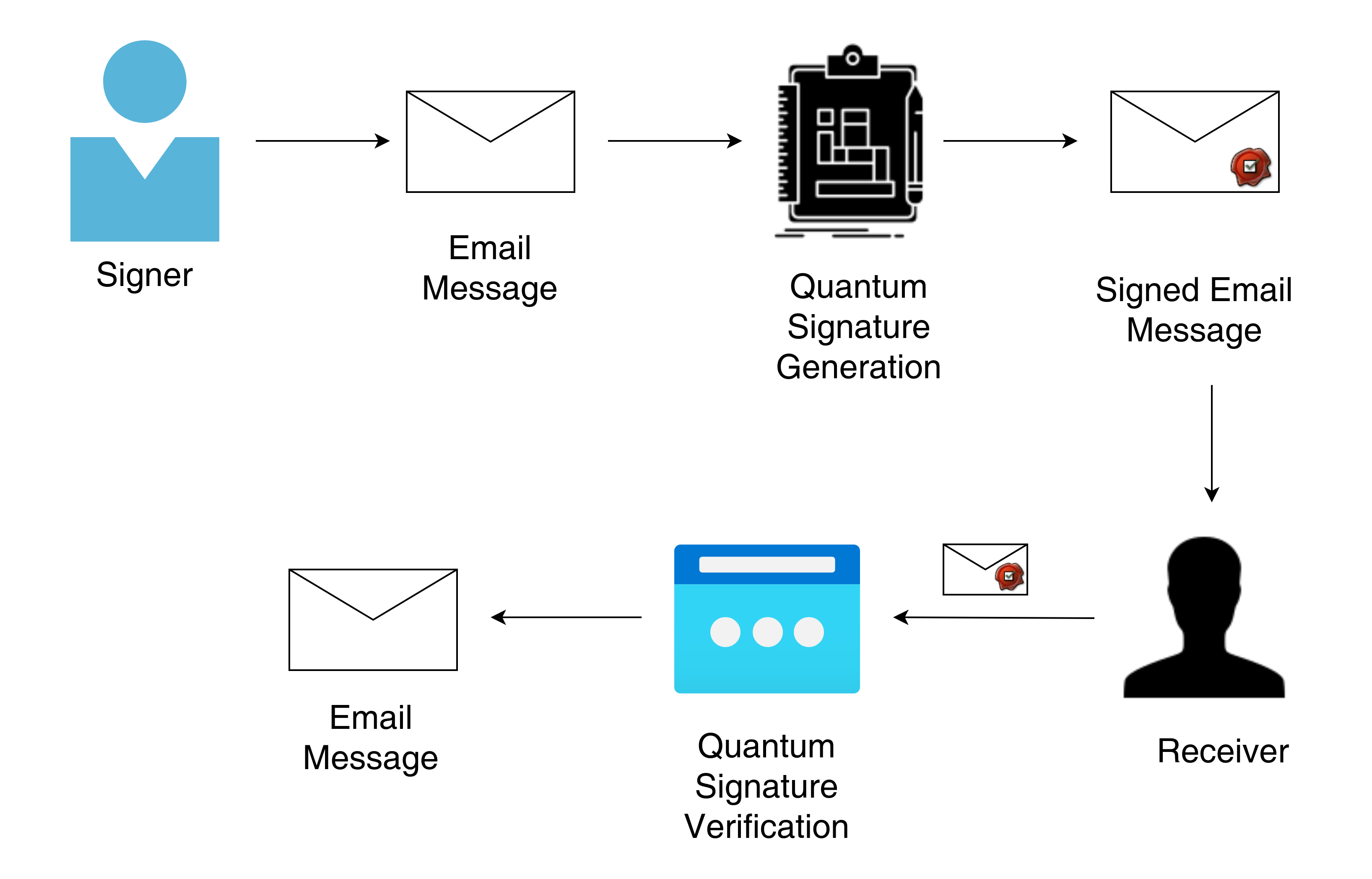}
    \caption{Application of quantum IBS in secure email}
    \label{fig:ibs-email}
\end{figure}

\section{Conclusion}
In this work, we used the quantum analogue of a one-time pad to design an information-theoretically secure quantum IBS. Furthermore, because our protocol's security is based on a fundamental aspect of quantum mechanics, it offers lasting safety compared to conventional and post-quantum IBS systems. Our proposed quantum IBS will satisfy unforgeability and undeniability if the encryption is theoretically secure. The unforgeability property ensures that no other party can do the signature on behalf of a signer. In contrast, undeniability ensures that a signer can not deny the signature generation if she did it. The suggested plan achieves long-term security and resists quantum attacks. Our scheme is more practical than the existing works since it uses only simple measurement operators and single-photon quantum resources. Furthermore, the suggested approach has a lower communication and computing overhead than the current quantum IBS systems \cite{chen2018public,xin2019identity,xin2020identity}. The Jupyter Notebook for the execution of the toy example is given in the Appendix. In addition, the application of our IBS in secure email communication is provided.

\noindent {\large \textbf{Statements and Declarations} } 

\noindent {\bf Competing Interests:} The authors declare no conflicts of interest.

\noindent {\bf Data Availability Statement:}
{Data will be available with the corresponding author upon reasonable request.}

\appendix
{


\section*{Supplementary material}

# QIBS2

May 17, 2022

# 1 Quantum Identity Based Signature(QIBS)

### 1.0.1 $ Message |p>=|010>, ID\_A ={011}, ID\_B={100}, T\_i= {011010}, T\_U={100101} = , U=U( /2, ,0)$

**Implentation on simulator and real quantum machine**

## 1.1 Signing a message m

```python
[1]: from qiskit import QuantumCircuit, Aer, BasicAer, transpile
     from qiskit.visualization import plot_histogram
     from qiskit.circuit import Gate
     from qiskit import *
     import numpy as np
     pi=np.pi
```

```python
[2]: pi=np.pi
     # Setting the message to |010>
     m=QuantumCircuit(3,3)
     m.x(1)
     m.barrier()
     # Applying U gate
     m.u(pi/2,pi,0,0)
     m.u(pi/2,pi,0,1)
     m.u(pi/2,pi,0,2)
     m.barrier()
     #Encryption
     m.z(0)
     m.x(1)
     m.x(2)
     m.barrier()
     m.draw()
```

[2]:

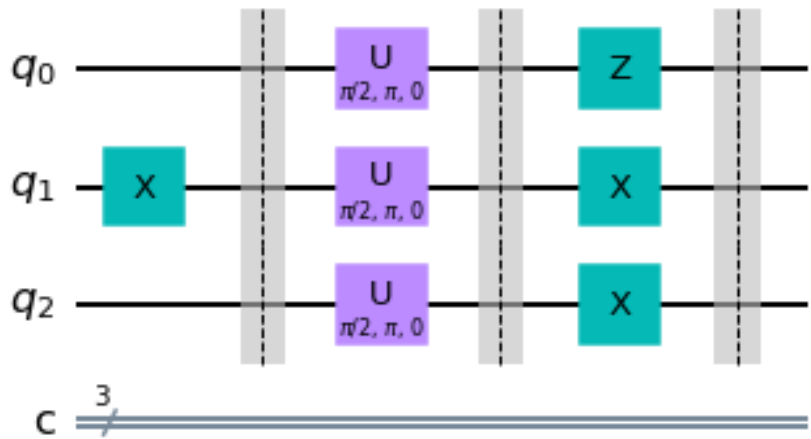

### 1.1.1 Alice sends $(|m>, |s>, |ID_i>)$ to BOB

## 1.2 Verification

BOB keeps |m> and encrypts $(|s>, |ID_i>)$ with $T\_U$ and sends to SKG

```
[3]: # Encryption of |s>
     m.barrier()
     m.x(0)
     m.x(1)
     m.z(2)
     m.barrier()
```

[3]: <qiskit.circuit.instructionset.InstructionSet at 0x21e5c557300>

[4]: `m.draw()`

[4]:

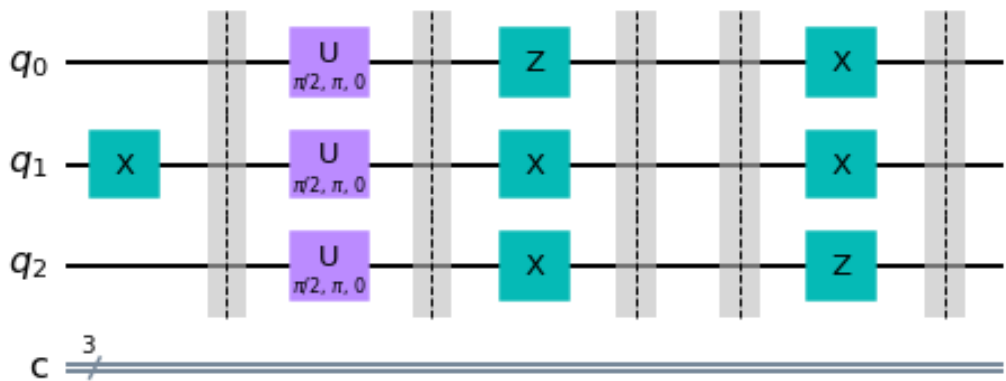

```
#SKG decrypts with Tu
m.x(0)
m.x(1)
m.z(2)
m.barrier()
m.draw()
```

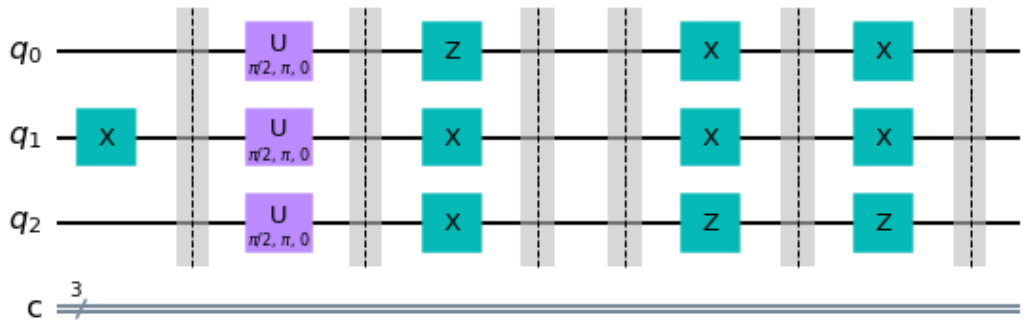

## 2 SKG Retrieves |m> by performing $DT_i$

```
# Decrypts by using Ti
m.z(0)
m.x(1)
m.x(2)
m.barrier()
m.draw()
```

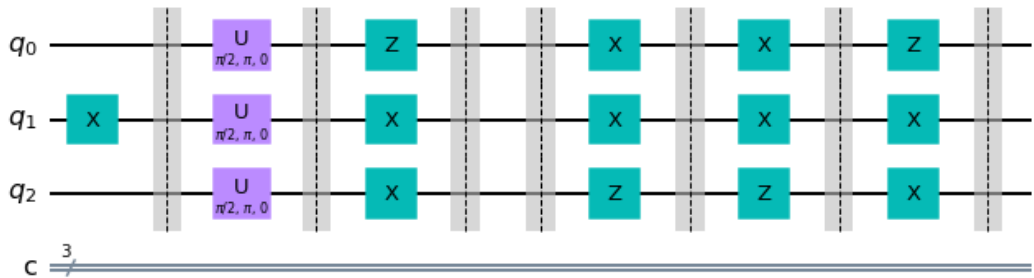

```
[7]: m.u(pi/2,pi,0,0).inverse()
     m.u(pi/2,pi,0,1).inverse()
     m.u(pi/2,pi,0,2).inverse()
     m.barrier()
     m.draw()
```

[7]:

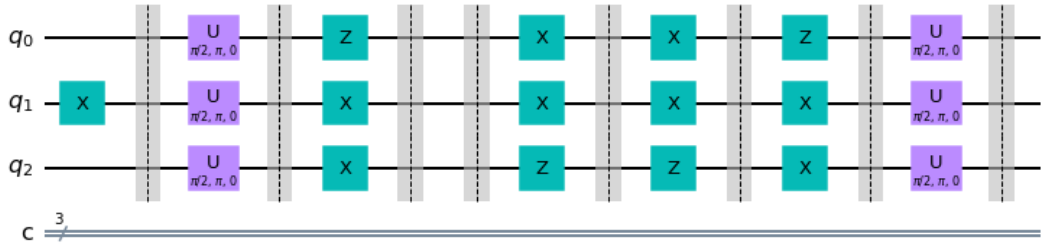

```
[8]: # Measure
     m.measure(0,0)
     m.measure(1,1)
     m.measure(2,2)
     m.draw()
```

[8]:

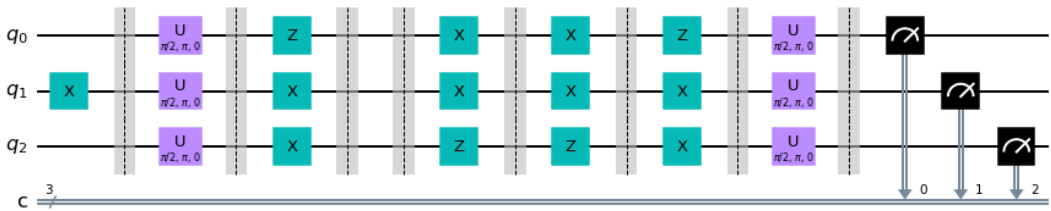

```
[9]: backend_sim=Aer.get_backend('qasm_simulator')
     job_sim=backend_sim.run(transpile(m,backend_sim),shots=1024)
     from qiskit.visualization import plot_histogram
     result_sim=job_sim.result()
     counts=result_sim.get_counts(m)
     print(counts)
```

    {'010': 1024}

### 2.0.1 SKG Encrypts using $T_u$ and sends to Bob

[ ]:

```
[10]: m.barrier()
      m.x(0)
      m.x(1)
      m.z(2)
      m.barrier()
      m.draw()
```

[10]:

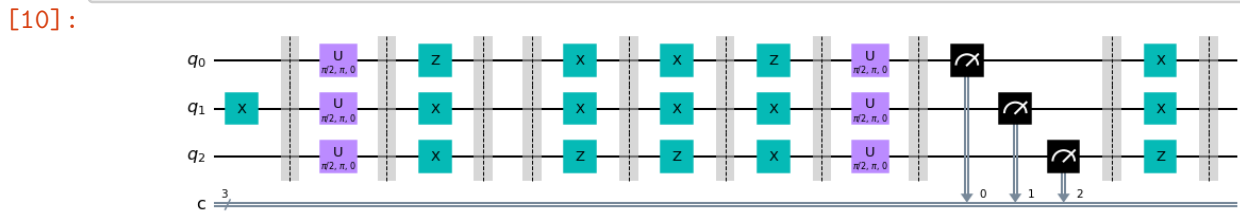

### 2.0.2 Bob decrypts it using $T_u$

```
[11]: m.x(0)
      m.x(1)
      m.z(2)
      m.barrier()
      m.draw()
```

[11]:

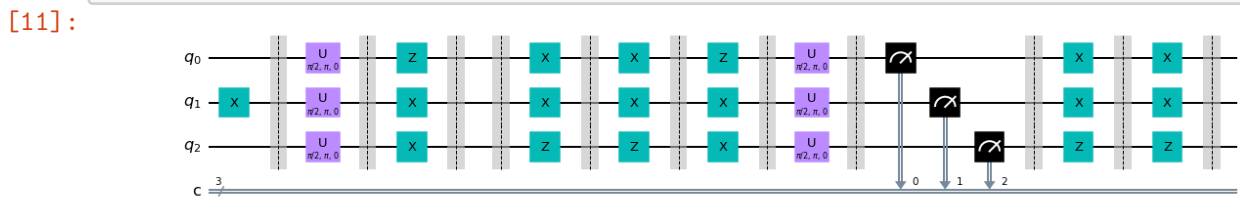

```
[12]: m.measure(0,0)
      m.measure(1,1)
      m.measure(2,2)
      m.draw()
```

[12]:

```
[13]: backend_sim=Aer.get_backend('qasm_simulator')
      job_sim=backend_sim.run(transpile(m,backend_sim),shots=1024)
      from qiskit.visualization import plot_histogram
      result_sim=job_sim.result()
      counts=result_sim.get_counts(m)
      print(counts)
```

{'010': 1024}

```
[14]: from qiskit.visualization import plot_histogram
      plot_histogram(counts)
```

[14]:

# 3 VERIFIED

### 3.0.1 Implementation on real quantum machine IBM Q Santiago

```python
[15]: from qiskit import IBMQ
```

```python
[17]: IBMQ.load_account()
```

```
[17]: <AccountProvider for IBMQ(hub='ibm-q', group='open', project='main')>
```

```python
[18]: provider=IBMQ.get_provider('ibm-q')
```

```python
[19]: available_cloud_backends = provider.backends()
      print('\nHere is the list of cloud backends that are available to you:')
      for i in available_cloud_backends: print(i)
```

```
Here is the list of cloud backends that are available to you:
ibmq_qasm_simulator
ibmq_armonk
ibmq_santiago
ibmq_bogota
ibmq_lima
ibmq_belem
ibmq_quito
simulator_statevector
simulator_mps
simulator_extended_stabilizer
simulator_stabilizer
ibmq_manila
```

```python
[21]: from qiskit.providers.ibmq import least_busy
      small_devices = provider.backends(filters=lambda x: x.configuration().n_qubits
        ↪== 5
                                        and not x.configuration().simulator)
      least_busy(small_devices)
```

```
[21]: <IBMQBackend('ibmq_belem') from IBMQ(hub='ibm-q', group='open', project='main')>
```

```python
[22]: qcomp=provider.get_backend('ibmq_santiago')
```

```python
[23]: job=execute(m, backend=qcomp)
```

```python
[24]: from qiskit.tools.monitor import job_monitor
```

```
[25]: job_monitor(job)
```

Job Status: job has successfully run

```
[26]: result=job.result()
```

```
[27]: plot_histogram(result.get_counts(m))
```

[27]:

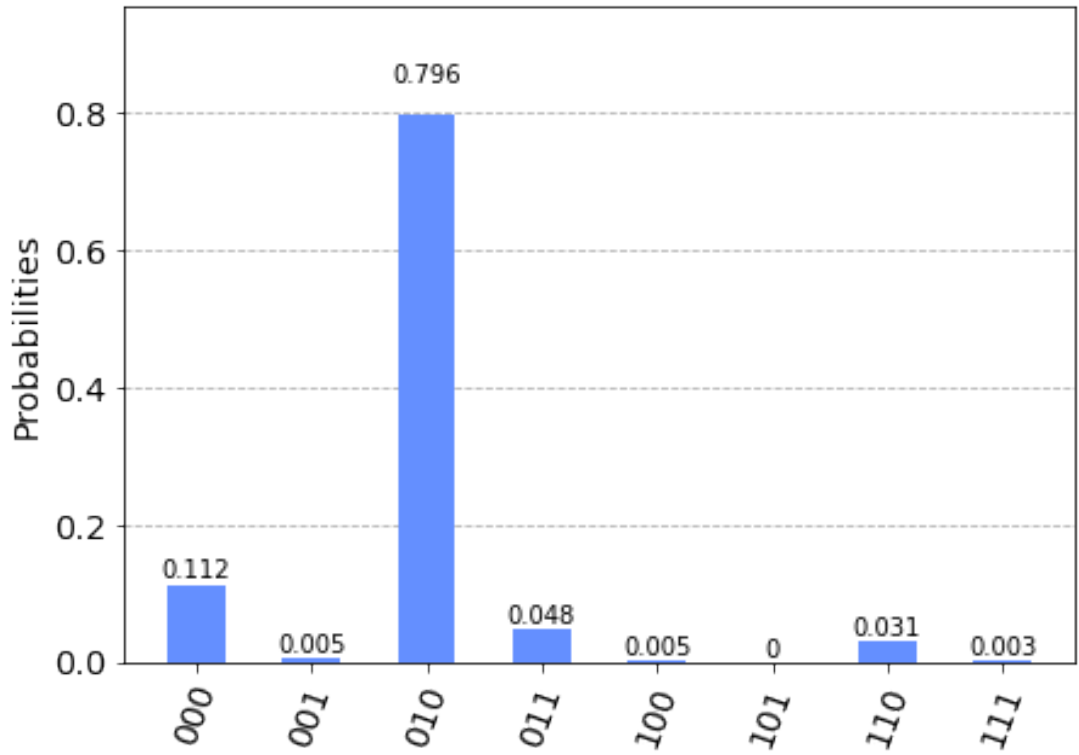

[ ]:


}


\begin{thebibliography}{99}
	\providecommand{\natexlab}[1]{#1}
	\providecommand{\url}[1]{\texttt{#1}}
	\expandafter\ifx\csname urlstyle\endcsname\relax
	\providecommand{\doi}[1]{doi: #1}\else
	\providecommand{\doi}{doi: \begingroup \urlstyle{rm}\Url}\fi
	
	\bibitem{an2019practical}
	An, Xue-Bi; Zhang, Hao; Zhang, Chun-Mei; Chen, Wei; Wang, Shuang; Yin,
	Zhen-Qiang; Wang, Qin; He, De-Yong; Hao, Peng-Lei; Liu, Shu-Feng, and others,
	.
	\newblock Practical quantum digital signature with a gigahertz bb84 quantum key
	distribution system.
	\newblock \emph{Optics letters}, 44\penalty0 (1):\penalty0 139--142, 2019.
	
	\bibitem{boykin2003optimal}
	Boykin, P~Oscar and Roychowdhury, Vwani.
	\newblock Optimal encryption of quantum bits.
	\newblock \emph{Physical review A}, 67\penalty0 (4):\penalty0 042317, 2003.
	
	\bibitem{chen2018public}
	Chen, Feng-Lin; Liu, Wan-Fang; Chen, Su-Gen, and Wang, Zhi-Hua.
	\newblock Public-key quantum digital signature scheme with one-time pad
	private-key.
	\newblock \emph{Quantum Information Processing}, 17\penalty0 (1):\penalty0 10,
	2018.
	
	\bibitem{james2019efficient}
	James, Salome and Reddy, P~Vasudeva.
	\newblock Efficient identity-based signature scheme with message recovery.
	\newblock In \emph{Journal of Physics: Conference Series}, volume 1344, page
	012016. IOP Publishing, 2019.
	
	\bibitem{ko2019forward}
	Ko, Hankyung; Jeong, Gweonho; Kim, Jongho; Kim, Jihye, and Oh, Hyunok.
	\newblock Forward secure identity-based signature scheme with rsa.
	\newblock In \emph{IFIP International Conference on ICT Systems Security and
		Privacy Protection}, pages 314--327. Springer, 2019.
	
	\bibitem{koblitz1987elliptic}
	Koblitz, Neal.
	\newblock Elliptic curve cryptosystems.
	\newblock \emph{Mathematics of computation}, 48\penalty0 (177):\penalty0
	203--209, 1987.
	
	\bibitem{kravitz1993digital}
	Kravitz, David~W.
	\newblock Digital signature algorithm, July~27 1993.
	\newblock US Patent 5,231,668.
	
	\bibitem{krzywiecki2019identity}
	Krzywiecki, {\L}ukasz; S{\l}owik, Marta, and Szala, Micha{\l}.
	\newblock Identity-based signature scheme secure in ephemeral setup and leakage
	scenarios.
	\newblock In \emph{International Conference on Information Security Practice
		and Experience}, pages 310--324. Springer, 2019.
	
	\bibitem{ramadan2020identity}
	Ramadan, Mohammed; Liao, Yongjian; Li, Fagen, and Zhou, Shijie.
	\newblock Identity-based signature with server-aided verification scheme for 5g
	mobile systems.
	\newblock \emph{IEEE Access}, 8:\penalty0 51810--51820, 2020.
	
	\bibitem{rivest1978method}
	Rivest, Ronald~L; Shamir, Adi, and Adleman, Leonard.
	\newblock A method for obtaining digital signatures and public-key
	cryptosystems.
	\newblock \emph{Communications of the ACM}, 21\penalty0 (2):\penalty0 120--126,
	1978.
	
	\bibitem{sahana2019provable}
	Sahana, Subhas~Chandra; Das, Manik~Lal, and Bhuyan, Bubu.
	\newblock A provable secure key-escrow-free identity-based signature scheme
	without using secure channel at the phase of private key issuance.
	\newblock \emph{S{\=a}dhan{\=a}}, 44\penalty0 (6):\penalty0 1--9, 2019.
	
	\bibitem{shamir1984identity}
	Shamir, Adi.
	\newblock Identity-based cryptosystems and signature schemes.
	\newblock In \emph{Workshop on the theory and application of cryptographic
		techniques}, pages 47--53. Springer, 1984.
	
	\bibitem{shi2022privacy}
	Shi, Run-hua and Li, Yi-fei.
	\newblock Privacy-preserving quantum protocol for finding the maximum value.
	\newblock \emph{EPJ Quantum Technology}, 9\penalty0 (1):\penalty0 1--14, 2022.
	
	\bibitem{shor1999polynomial}
	Shor, Peter~W.
	\newblock Polynomial-time algorithms for prime factorization and discrete
	logarithms on a quantum computer.
	\newblock \emph{SIAM review}, 41\penalty0 (2):\penalty0 303--332, 1999.
	
	\bibitem{song2020efficient}
	Song, Dawei and Wen, Fengtong.
	\newblock Efficient identity-based signature authentication scheme for smart
	home system.
	\newblock In \emph{International Conference on Artificial Intelligence and
		Security}, pages 639--648. Springer, 2020.
	
	\bibitem{ullah2020lightweight}
	Ullah, Syed~Sajid; Ullah, Insaf; Khattak, Hizbullah; Khan, Muhammad~Asghar;
	Adnan, Muhammad; Hussain, Saddam; Amin, Noor~Ul, and Khattak, Muazzam A~Khan.
	\newblock A lightweight identity-based signature scheme for mitigation of
	content poisoning attack in named data networking with internet of things.
	\newblock \emph{IEEE Access}, 8:\penalty0 98910--98928, 2020.
	
	\bibitem{wang2020efficient}
	Wang, Chang-Ji; Huang, Hui, and Yuan, Yuan.
	\newblock Efficient pairing-free provably secure scalable revocable
	identity-based signature scheme.
	\newblock \emph{Journal of Internet Technology}, 21\penalty0 (2):\penalty0
	503--509, 2020.
	
	\bibitem{wei2017forward}
	Wei, Jianghong; Liu, Wenfen, and Hu, Xuexian.
	\newblock Forward-secure identity-based signature with efficient revocation.
	\newblock \emph{International Journal of Computer Mathematics}, 94\penalty0
	(7):\penalty0 1390--1411, 2017.
	
	\bibitem{wu2020leakage}
	Wu, Jui-Di; Tseng, Yuh-Min; Huang, Sen-Shan, and Tsai, Tung-Tso.
	\newblock Leakage-resilient revocable identity-based signature with cloud
	revocation authority.
	\newblock \emph{Informatica}, 31\penalty0 (3):\penalty0 597--620, 2020.
	
	\bibitem{xin2015quantum}
	Xin, Xiangjun and Li, Fagen.
	\newblock Quantum authentication of classical messages without entangled state
	as authentication key.
	\newblock \emph{International Journal of Multimedia and Ubiquitous
		Engineering}, 10\penalty0 (8):\penalty0 199--206, 2015.
	
	\bibitem{xin2019identity}
	Xin, Xiangjun; Wang, Zhuo, and Yang, Qinglan.
	\newblock Identity-based quantum signature scheme with strong security.
	\newblock \emph{Optical and Quantum Electronics}, 51\penalty0 (12):\penalty0
	1--13, 2019.
	
	\bibitem{xin2020identity}
	Xin, Xiangjun; Wang, Zhuo, and Yang, Qinglan.
	\newblock Identity-based quantum signature based on bell states.
	\newblock \emph{Optik}, 200:\penalty0 163388, 2020{\natexlab{a}}.
	
	\bibitem{xin2020efficient}
	Xin, Xiangjun; Wang, Zhuo; Yang, Qinglan, and Li, Fagen.
	\newblock Efficient identity-based public-key quantum signature scheme.
	\newblock \emph{International Journal of Modern Physics B}, 34\penalty0
	(10):\penalty0 2050087, 2020{\natexlab{b}}.
	
	\bibitem{yin2016practical}
	Yin, Hua-Lei; Fu, Yao, and Chen, Zeng-Bing.
	\newblock Practical quantum digital signature.
	\newblock \emph{Physical Review A}, 93\penalty0 (3):\penalty0 032316, 2016.
	
	\bibitem{zhang2020practical}
	Zhang, Chun-Mei; Zhu, Yan; Chen, Jing-Jing, and Wang, Qin.
	\newblock Practical quantum digital signature with configurable decoy states.
	\newblock \emph{Quantum Information Processing}, 19\penalty0 (5):\penalty0
	1--7, 2020.
	
	\bibitem{zhang2019high}
	Zhang, Hao; An, Xue-Bi; Zhang, Chun-Hui; Zhang, Chun-Mei, and Wang, Qin.
	\newblock High-efficiency quantum digital signature scheme for signing long
	messages.
	\newblock \emph{Quantum Information Processing}, 18\penalty0 (1):\penalty0
	1--9, 2019.
	
	\bibitem{zhao2019communication}
	Zhao, Jing; Wei, Bin, and Su, Yang.
	\newblock Communication-efficient revocable identity-based signature from
	multilinear maps.
	\newblock \emph{Journal of Ambient Intelligence and Humanized Computing},
	10\penalty0 (1):\penalty0 187--198, 2019.
 \bibitem{bibid},
  Kitak won ; Jino Heo ; Chun Seok Yoon ;  Ji-Woong Choi and  Hyung-Jin Yang. 
  \newblock Quantum Signature Scheme for Participant Attack.
  
  \newblock Journal of the Korean Physical Society \penalty0:
75, 
  271--276,
  2019,
  Springer.

	
\end{thebibliography}
\end{document}